\documentclass[conference,twocolumn]{IEEEtran}
\IEEEoverridecommandlockouts
\usepackage{amssymb}
\usepackage{amsmath}
\usepackage{graphicx,subfigure}
\usepackage{amsfonts}
\usepackage{color}
\usepackage{ dsfont }
\newtheorem{theorem}{Theorem}

\newtheorem{lemma}{Lemma}
\newtheorem{corollary}{Corollary}
\newtheorem{definition}{Definition}
\newtheorem{remark}{Remark}
\newtheorem{example}{Example}

\newcommand{\beqno}{ \begin{equation*} }
\newcommand{\eeqno}{ \end{equation*} }
\newcommand{\beq}{ \begin{equation} }
\newcommand{\eeq}{ \end{equation} }

\newcommand{\calM}{\mathcal{M}}
\newcommand{\calC}{\mathcal{C}}
\newcommand{\bbf}{\mathbb{F}}

\begin{document}

\title{Repairable Replication-based Storage Systems Using Resolvable Designs}

\author{\IEEEauthorblockN{Oktay Olmez}
\IEEEauthorblockA{Department of Mathematics,\\
Iowa State University,\\
Ames, Iowa 50011. \\
Email: oolmez@iastate.edu}
\and
\IEEEauthorblockN{Aditya Ramamoorthy}
\IEEEauthorblockA{Department of Electrical and Computer Engineering,\\
Iowa State University,\\
Ames, Iowa 50011. \\
Email: adityar@iastate.edu}
\thanks{This work was supported in part by NSF grant CCF-1018148.}
}

\maketitle
\begin{abstract}
We consider the design of regenerating codes for distributed storage systems at the minimum bandwidth regeneration (MBR) point. The codes allow for a repair process that is exact and uncoded, but table-based. These codes were introduced in prior work and consist of an outer MDS code followed by an inner fractional repetition (FR) code where copies of the coded symbols are placed on the storage nodes. The main challenge in this domain is the design of the inner FR code.

In our work, we consider generalizations of FR codes, by establishing their connection with a family of combinatorial structures known as resolvable designs. Our constructions based on affine geometries, Hadamard designs and mutually orthogonal Latin squares allow the design of systems where a new node can be exactly regenerated by downloading $\beta \geq 1$ packets from a subset of the surviving nodes (prior work only considered the case of $\beta = 1$). Our techniques allow the design of systems over a large range of parameters. Specifically, the repetition degree of a symbol, which dictates the resilience of the system can be varied over a large range in a simple manner. Moreover, the actual table needed for the repair can also be implemented in a rather straightforward way. Furthermore, we answer an open question posed in prior work by demonstrating the existence of codes with parameters that are not covered by Steiner systems.

\end{abstract}
\section{Introduction}
Large scale data storage systems are becoming ubiquitous in recent years. The availability of low cost storage media such as magnetic disks have fueled the growth of various applications such as Facebook, Youtube etc. These applications require a massive amount of data to be stored and accessed in a distributed manner. In these systems it is often the case that the individual storage nodes are unreliable. Thus, the integrity of the data and the speed of the data access needs to be maintained even under the presence of such unreliable storage nodes. This issue is typically handled by introducing redundancy in the storage system. For instance, one could replicate data across multiple nodes or use  Maximum Distance Separable (MDS) codes such as Reed-Solomon codes that allow for a better reliability at the same redundancy.

However, the {\it large scale distributed} nature of the systems under consideration introduces another issue. Namely, if a given storage node fails, it need to be regenerated so that the new system continues to have the properties of the original system. It is of course desirable to perform this regeneration in a distributed manner and download as little data as possible from the existing nodes. The problem of regenerating codes was introduced by Dimakis et al. \cite{Dim}. The authors demonstrated a fundamental tradeoff between the amount of data stored at each node (storage capacity) and the amount of data that needs to be downloaded for regenerating a failed node (repair bandwidth).

In particular, consider a distributed storage system (DSS) that consists of $n$ storage nodes, each of which stores $\alpha$ packets. A given user needs to have the ability to reconstruct the stored file by contacting any $k$ nodes; this is referred to as the MDS property of the system. Suppose that a given node fails. The DSS needs to be repaired by introducing a new node. This node should be able to contact any $d \geq k$ surviving nodes and download $\beta$ packets from each of them for a total repair bandwidth of $\gamma = d \beta$ packets. The new DSS should continue to have the MDS property. The work of \cite{Dim} considered the case of {\it functional repair}, where the new node needs to be functionally equivalent to the failed node. It was shown that this could be achieved by the usage of random network coding. In particular, under functional repair the entire storage vs. repair bandwidth curve is known exactly. One can also consider {\it exact repair} where the new node should be able to recreate the contents of the failed node (see \cite{RSKR09, RSKR10}). Two points on the curve deserve special mention and are arguably of most interest from a practical perspective. The minimum bandwidth regenerating (MBR) point refers to the point where the repair bandwidth, $\gamma$ is minimum. Likewise, the minimum storage regenerating (MSR) point refers to the point where the storage per node is minimized.

Much of the existing work in the area of DSS considers {\it coded} repair where the surviving nodes need to compute linear combinations of all their existing packets. It is well recognized that the read/write bandwidth of machines is much lower than the network bandwidth. Thus, this process induces undesirable latencies in the repair process. The process can also be potentially memory intensive if the packets comprising the file are large. Motivated by these issues, in \cite{RR}, El Rouayheb and Ramchandran considered the following variant of the DSS problem. The DSS needs to satisfy the property of {\it exact} and {\it uncoded} repair, i.e., the regenerating node needs to produce an exact copy of the failed node by simply downloading packets from the surviving nodes. This allows the entire system to work without requiring any computation at the surviving nodes. In addition they considered systems that are resilient to multiple $(> 1)$ failures. However, the DSS only has the property that the repair can be conducted by contacting some set of $d$ nodes, i.e., unlike the original setup, repair is not guaranteed by contacting any set of $d$ nodes. This is reasonable as most practical systems operate via a table-based repair, where the new node is provided information on the set of surviving nodes that it needs to contact. The work of \cite{RR} proposed a construction whereby an outer MDS code is concatenated with an inner ``fractional repetition" code of a certain degree. The main challenge here is to design the inner fractional repetition code in a systematic manner.

The work of \cite{RR} primarily considered fractional repetition (FR) codes that result from Steiner systems, which are an instance of a combinatorial design. Subsequently, Koo and Gill \cite{kooG11} considered the usage of finite projective planes for the design of these codes.  Both \cite{RR} and \cite{kooG11}, consider fractional repetition codes where the new node downloads exactly one packet (i.e., $\beta = 1$) from the surviving nodes that are contacted. In this work we study the design of fractional repetition codes in more generality.
\subsection{Main Contributions}
In this work we consider \underline{REP}airable \underline{RE}plication-based \underline{S}torage \underline{S}ystems \underline{U}sing \underline{RE}solvable \underline{D}esigns, abbreviated as REPRESSURED codes.
REPRESSURED codes are more general than fractional repetition codes as the new node has the flexibility of downloading $\beta \geq 1$ packets from the surviving nodes. 
Our design is based on combinatorial structures called resolvable designs \cite{St}. Our work makes the following contributions.
\begin{itemize}
\item Our constructions based on affine geometries and Hadamard designs allow for a large class of codes where $\beta \geq 1$.
\item The work of \cite{RR} considers Steiner systems where parameters such as the repetition degree of each packet are fixed a priori. In contrast, our code design allows the system designer to vary the repetition degree within a large range in a simple manner.
\item We resolve an open question posed in \cite{RR}, by showing the existence of FR codes that have a repetition degree greater than two, that cannot be constructed by Steiner systems.
\item The systems under consideration require table-based repair, whereby a table of nodes that need to be contacted under the various failure patterns needs to be maintained. As will be evident, our code design approach is such that this table can be maintained in a very simple manner.
\end{itemize}
This paper is organized as follows. Section \ref{sec:problem_form} contains a formal discussion of the problem formulation. Section \ref{sec:affine_resolv} and Section \ref{sec:general_resolv} discuss the design of FR codes from resolvable designs and Latin squares respectively. We conclude the paper with a comparison with existing work and discussion of future issues in Section \ref{sec:conc_remarks}.
\section{Problem Formulation}
\label{sec:problem_form}
The DSS is specified by parameters $(n,k,d)$ where $n$ - number of storage nodes, $k$ - number of nodes to be contacted for recovering the file and $d \geq k$ is the number of nodes to be contacted in order to regenerate a failed node. The storage capacity of each node is denoted by $\alpha$. In case of repair, the new node downloads $\beta$ packets from each surviving node, for a total of $\gamma = d \beta$ bits. Let $\calM$ denote the size of file being stored on the DSS. Under functional repair, it is known that at the MBR point, $\alpha = \gamma = \frac{2 \calM d}{2kd - k^2 + k}$.

We consider the design of fractional repetition codes that are best explained by means of the following example \cite{RR} with $(n,k,d) = (5,3,4)$ in the discussion below.  
\begin{example}

\begin{figure} [ht]
\centering
\includegraphics[scale=0.45]{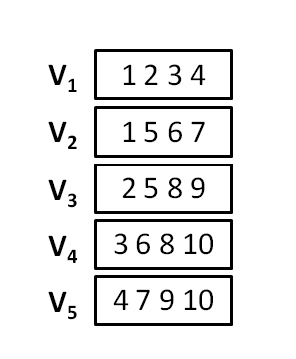}
\caption{A DSS with $(n,k,d) = (5,3,4)$. Each node contains a subset of the packets from $\{y_1, \dots, y_{10}\}$. Node $V_1$ for instance contains symbols $y_i, i = 1, \dots, 4$.}
\label{DSS-(5,3,4)}
\end{figure}

Consider a file of $\calM = 9$ packets $(x_1, \dots, x_9) \in \mathbb{F}_q^9$ that needs to stored on the DSS. We use a $(10,9)$ MDS code that outputs $10$ packets $y_i = x_i, i = 1, \dots, 9$ and $y_{10} = \sum_{i=1}^9 x_i$. The coded packets $y_1, \dots, y_{10}$ are placed on $n=5$ storage nodes as shown in Fig. \ref{DSS-(5,3,4)}. This placement specifies the inner fractional repetition code. It can be observed that each $y_i$ is repeated $\rho = 2$ times and the total number of symbols $\theta = 10$. Any user who contacts any $k=3$ nodes can recover the file (using the MDS property). Moreover, it can be verified that if a node fails, one packet each can be downloaded from the four surviving nodes, i.e., $\beta =1 $ and $d = 4$, so that $\gamma = 4$.
\end{example}
Thus, the approach uses an MDS code to encode a file consisting of a certain number of symbols. Let $\theta$ denote the number of encoded symbols. Copies of these symbols are placed on the $n$ nodes such that each symbol is repeated $\rho$ times and each node contains $\alpha$ symbols. Moreover, if a given node fails, it can be exactly recovered by downloading $\beta$ packets from some set of $d$ surviving nodes, for a total repair bandwidth of $\gamma=d\beta$. It is to be noted that in this case $\alpha = \gamma$, i.e., these schemes operate at the MBR point. In the example above, $\beta = 1$, so that $\alpha = d$. However, one can consider systems with $\beta > 1$ in general.

In this work we propose the construction of several fractional repetition codes. 
Before introducing the formal definition of a fractional repetition (FR) code we need the notion of $\beta$-recoverability. Let $[n]$ denote the set $\{1, 2, \dots, n\}$.
\begin{definition}
Let $\Omega = [\theta]$ and $V_i, i = 1, \dots, d$ be subsets of $\Omega$. Let $V = \{V_1, \dots, V_d\}$ and consider $A \subset \Omega$ with $|A| = d\beta$. We say that $A$ is $\beta$-recoverable from $V$ if there exists $B_i \subseteq V_i$ for each $i= 1, \dots, d$ such that $B_i \subset A, |B_i| = \beta$ and $\displaystyle \cup_{i=1}^d B_i = A$. 

\end{definition}
\begin{definition}
\label{defn:fr_code}
A fractional repetition (FR) code $\calC = (\Omega, V)$ for a $(n,k,d)$ DSS (where $d \geq k$) with repetition degree $\rho$ and normalized repair bandwidth $\beta = \alpha/d$  ($\alpha$ and $\beta$ are positive integers) is a set of $n$ subsets $V=\{V_1, \dots, V_n\}$ of a symbol set $\Omega = [\theta]$ with the following properties.
\begin{itemize}
\item[(a)] The cardinality of each $V_i$ is $\alpha$.
\item[(b)] Each element of $\Omega$ belongs to $\rho$ sets in $V$.
\item[(c)] Let $V^{surv}$ denote any 
$(n- \rho_{res})$ sized subset of $V$ and $V^{fail} = V \setminus V^{surv}$. Each $V_j \in V^{fail}$ is $\beta$-recoverable from some $d$-sized subset of $V^{surv}$.
\end{itemize}
The value of $\rho_{res}$ is a measure of the resilience of the system to node failures.
The code rate is defined as
\begin{align*}
\displaystyle R_{\calC}(k) &= \min_{I \subset [n],|I| = k} |\cup_{i \in I} V_i|,
\end{align*}
where $[n] = \{1, \dots, n\}$.
\end{definition}
It can be observed that $R_{\calC}(k)$ corresponds to the maximum filesize that can be obtained with a certain value of $k$. We remark that in \cite{RR}, only FR codes with $\beta = 1$ were studied. In this case the requirement (c) in Definition \ref{defn:fr_code} is automatically satisfied and it can be seen that the system is resilient to $\rho - 1$ failures.

Our proposed constructions aim to maximize the file size $\calM$ given the parameters $(n, k, d, \alpha)$\footnote{It can be seen that these further specify $\beta = \alpha/d$. Furthermore, it can be seen that $n d = \theta \rho$.} while ensuring a certain level of failure resilience. In the case of MBR constructions the maximum filesize under the Dimakis et al. model is known to be $k\alpha - \beta \binom{k}{2}$. Accordingly, we call a FR code {\it universally good} if the code rate $R_{\calC}(k) \geq k\alpha - \binom{k}{2} \beta$ (\cite{RR} used this terminology when $\beta = 1$). As will be evident, all our constructions in this work are universally good.



\section{REPRESSURED codes from Affine Resolvable Designs}
\label{sec:affine_resolv}
We now discuss the construction of FR codes from resolvable designs. As we shall see this construction allows us to easily vary the repetition degree $\rho$ and the normalized repair bandwidth $\beta$.
\begin{definition}\label{resolvable fractional repetition code} Let $\calC = (\Omega, V)$ where $V = \{V_1, \dots, V_n\}$ be a FR code. A subset $P \subset V$ is said to be a parallel class if for $V_i \in P$ and $V_j \in P$ with $i \neq j$ we have $\displaystyle V_i \cap V_j = \emptyset$ and $\cup_{\{j : V_j \in P\}} V_j = \Omega$. A partition of  $V$ into $r$ parallel classes is called a resolution. If there exists at least one resolution then the code is called a resolvable fractional repetition code.
\end{definition}
The properties of a resolvable FR code are best illustrated by means of the following example.
\begin{example}
Consider a DSS construction with parameters $\alpha = 3, \theta = \alpha^2 = 9, \rho = 2$ and $\beta = 1$. Suppose that we arrange the symbols in $\Omega$ in a $\alpha \times\alpha$ array $A$ shown below.
$$A=\begin{array}{ccc}
1&2&3\\
4&5&6\\
7&8&9
\end{array}$$
Let the rows and the columns of $A$ form the nodes in the FR code $\calC$ (see Fig. \ref{DSS-(6,3,3)}), thus $n=6$. It is evident that there are two parallel classes in $\calC$, $P^r = \{V_1, V_2, V_3\}$ (corresponding to rows) and $P^c = \{V_4, V_5, V_6\}$ (corresponding to columns). As $\rho = 2$, this code can tolerate one failure.
\begin{figure} [t]
\centering
\includegraphics[scale=0.45]{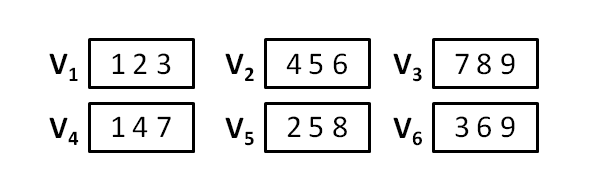}
\caption{A (6,3,3) REPRESSURED code. Note that $\{V_1, V_2, V_3\}$ and $\{V_4, V_5, V_6\}$ form parallel classes.}
\label{DSS-(6,3,3)}
\end{figure}
It can be observed that for $V_i \in P^r$ and $V_j \in P^c$, we have $|V_i \cap V_j| = 1$. Using this we can compute the code rate when $k=3$, $R_{\calC}(3)$ as follows. Let $a + b=3$ with $a \geq b$. Then, the number of distinct symbols in a set of $3$ nodes from $\calC$ is
\begin{equation*}
3a+(3-a)(3-a)= a^2 + 9 - 3a,
\end{equation*}
where $a$ nodes are from $P^r$ and $(3-a)$ nodes are from $P^c$. This is minimized when $a = 2$. Thus, $R_{\calC}(3) = 7$ and it can be seen that the construction is universally good.
\end{example}


It can be seen that given a resolvable FR code $\calC = (\Omega, V)$ with $r$ parallel classes, one can obtain a resolvable FR code $\calC'$ with repetition degree $\rho \leq r$, simply by choosing the node set in $\calC'$ to be any $\rho$ distinct parallel classes from $\calC$. Moreover, the recovery process when at most $\rho -1$ nodes are in failure and $\beta = 1$ is also quite simple. Specifically, it is clear that upon $\rho -1$ node failures, there is at least one parallel class in $\calC'$ that remains intact. As all symbols from $\Omega$ are represented in any parallel class, any failed node can be regenerated by contacting the nodes in the remaining class.

%
%

We now present explicit constructions of resolvable FR codes by leveraging the properties of combinatorial designs. For an in-depth  discussion of combinatorial designs, see \cite{St}. 
\begin{definition}
A $(\theta,\rho, \alpha,\lambda)$ balanced incomplete block design (BIBD) is
a pair $(\Omega, V )$, where $\Omega$ is a $\theta$-element set and $V$ is a
collection of $\alpha$-subsets of $\Omega$, called blocks, such that $|V| = n$; every element of $V$ is contained in exactly $\rho$ blocks and every $2$-subset of $\Omega$ is contained in exactly $\lambda$ blocks.
\end{definition}
Let $n$ denote the number of blocks. It can be shown that for a BIBD, the following relations hold.
\begin{align}
n \alpha &= \theta \rho, \label{bibd_eq_1}\\
\rho(\alpha - 1) &= \lambda(\theta -1). \label{bibd_eq_2}
\end{align}

It can be observed that a BIBD is essentially a FR code, with the additional property that every $2$-subset of $\Omega$ is contained in exactly $\lambda$ blocks. Likewise we can define a resolvable $(\theta,\rho,\alpha,\lambda)$-BIBD (analogous to a resolvable FR code) and the notions of a parallel class and resolution. Namely, a parallel class is a subset of disjoint blocks from $V$ whose union is $\Omega$ and a partition of $V$ into $\rho$ parallel classes is a resolution.
\begin{definition}
A $S(t, \alpha, \theta)$ Steiner system is a set $\Omega$ of $\theta$ elements and a collection of subsets of $\Omega$ of size $\alpha$ called blocks such that any $t$ subset of the symbol set $\Omega$ appears in exactly one of the blocks.
\end{definition}
It can be seen that a $S(2, \alpha, \theta)$ Steiner system is a $(\theta,\rho,\alpha,1)$-BIBD where $\displaystyle \rho=\frac{\theta-1}{\alpha-1}$.

\begin{lemma} {\it Bose's Inequality} \cite{Bose1}. Suppose that there exists a resolvable $(\theta,\rho,\alpha,\lambda)$-BIBD. Then, $n \geq \theta + \rho -1$.
\end{lemma}

Within the class of resolvable designs we will primarily be interested in the class of affine resolvable designs for which $n = \theta + \rho - 1$.

\subsection{Affine geometry based constructions}
First, we will discuss the construction of a resolvable $(q^2, q+1, q,1)$-BIBD. This is also known as the affine plane of order $q$.


We can explicitly construct affine planes when $q$ is a prime power. Let $q$ be a prime power and $\mathbb{F}_q$ denote the finite field of order $q$. We define the symbol set $\Omega = \mathbb{F}_q \times \mathbb{F}_q$. For any $a,b \in \mathbb{F}_q $, define a block  $V_{a,b}=\{(x,y)\in \Omega: y=ax+b\}$. For any $c \in \mathbb{F}_q$, define  $V_{\infty,c}=\{(c,y)\in \Omega: y \in\mathbb{F}_q \}$. So there are $q^2+q$ blocks which we can partition into $q+1$ parallel classes of each size $q$. Specifically, fix $a \in \mathbb{F}_q$ then $\{V_{a,b}: b\in \mathbb{F}_q\}$ forms the $q$ parallel classes and the last parallel class is given by  $\{V_{\infty,c}: c\in \mathbb{F}_q\}$.
\begin{figure} [t]
\centering
\includegraphics[scale=0.20]{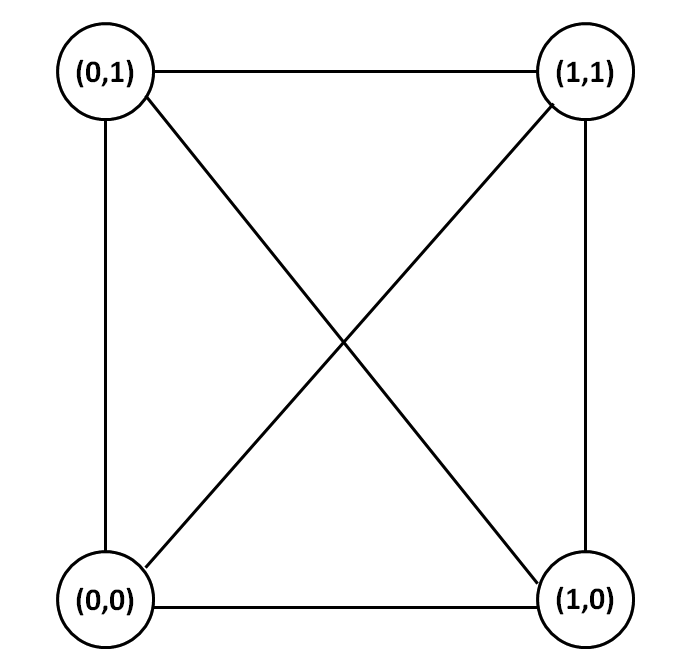}
\caption{A (4, 2, 2) REPRESSURED code from the affine plane of order 2. Each storage node consists of symbols that lie on a line, e.g., $\{(0,0), (0,1)\}$ is a storage node.}
\label{order2}
\end{figure}
\begin{example} By using the above construction we can construct an affine plane of order 2 (see Fig. \ref{order2}).

The set of symbols is $\Omega = \mathbb{F}_2 \times \mathbb{F}_2$ and the blocks are as follows:
$$V_{0,0}=\{(0,0),(1,0)\}$$
$$V_{0,1}=\{(0,1),(1,1)\}$$
$$V_{1,0}=\{(0,0),(1,1)\}$$
$$V_{1,1}=\{(0,1),(1,0)\}$$
$$V_{\infty,0}=\{(0,0),(0,1)\}$$
$$V_{\infty,1}=\{(1,0),(1,1)\}$$
\end{example}

Affine planes are also considered in \cite{RR} since any affine plane is a Steiner system. They also mentioned in \cite{kooG11}. However, here we have the flexibility of constructing fractional repetition codes with repetition $\rho \leq q+1$ by choosing any $\rho$ parallel classes. For instance, Example 2 can be also constructed by considering only two parallel classes of an affine plane of order $\alpha$ when $\alpha$ is a prime power.

Next, we discuss affine geometries which yield a larger class of constructions. Let $q$ be a prime power, $m \geq 2$ and $\Omega = \mathbb{F}_q^m$. Let $1 \leq \delta \leq m-1$. Note that $\Omega$ is an $m$-dimensional vector space over $\mathbb{F}_q$. A $\delta$-flat is a solution set to a system of $m- \delta$ independent linear equations that can be homogeneous or non-homogeneous. The set $\Omega$ and the set of all $\delta$-flats of $\Omega$ comprise the $m$-dimensional affine geometry over $\mathbb{F}_q$, denoted by $AG_m(q)$. It turns out that one can generate a large class of resolvable designs by considering $AG_m(q)$. Let ${m \brack \delta}_q$ denote the Gaussian coefficient, so that
\begin{align*}
{m \brack \delta}_q = \begin{cases} \frac{(q^m-1) (q^{m-1} - 1) \dots (q^{m-\delta+1} - 1)}{(q^\delta -1)(q^{\delta-1} - 1) \dots (q-1)} & \text{~if~} \delta \neq 0\\
1 &  \text{~if~} \delta = 0
\end{cases}
\end{align*}
\begin{theorem}
\label{thm:resolvable_bibd} \cite{St} Let $V$ denote the set of all $\delta$-flats in $AG_m(q)$. Then $\Omega = \mathbb{F}_q^m$ and $V$ form a resolvable $(q^m, \rho, q^\delta,\lambda)$ design with $n = q^{m-\delta} {m \brack \delta}_q, \rho = {m \brack \delta}_q$ and $\lambda = {m-1 \brack \delta-1}_q$.
\end{theorem}
The case of $m=2, \delta = 1$ corresponds to the case of affine planes that were discussed above. It can be shown that when $\delta = m-1$, using Theorem \ref{thm:resolvable_bibd}, we obtain affine resolvable designs with $n = \theta + \rho - 1$. In this case the DSS parameters are $\theta=q^m$, $\alpha=q^{m-1}$, $\rho=\frac{q^{m}-1}{q-1}$ and $n=q\rho$.
The design can be specified by means of the following algorithm.
\begin{itemize}
\item[(i)] Let $\Omega = \{(x_1,x_2, \cdots, x_m): x_i \in \mathbb{F}_q~\mbox{for}~i=1,2,\cdots, m \}$ be the symbol set. 
\item[(ii)] Find $\rho$, $(m-1)$-dimensional subspaces of $\bbf_q^m$ such that each of them contains the symbol $(0,0,\cdots, 0) \in \mathbb{F}_q^m$. Note that these subspaces of $\mathbb{F}_q^m$ are the solutions to homogeneous linear equations over $\mathbb{F}_q$ in $q$ variables. These $\rho$ subspaces are representatives of the $\rho$ different parallel classes.
\item[(iii)] Construct each parallel class by considering the additive cosets of its representative.
Let $R_1$ be a $(m-1)$-dimensional subspace corresponding a given homogenous equation. Take a symbol $u_1 \notin R_1$ and add $u_1$ to each symbol in $R_1$ to form a subspace $R_1'$ (which corresponds to a nonhomogeneous equation). We note that there are $q-1$ non-zero choices for $u_1$. Each choice of $u_1$ forms a new block. 

\end{itemize}
\begin{example} \cite{St} Let $q=3$ and $m=3$. The set of symbols is $\Omega = \mathbb{F}_3^3$ and there are 39 blocks which can be partitioned into 13 parallel classes. The representatives of the 13 parallel classes are as follows:
$$R_1 = \{000, 001, 002, 010, 020, 011, 012, 021, 022\}$$
$$ R_2 =\{000, 001, 002, 100, 200, 101, 102, 201, 202\} $$
$$R_3 = \{000, 001, 002, 110, 220, 111, 112, 221, 222\} $$
$$R_4 = \{000, 001, 002, 120, 210, 121, 122, 211, 212\}$$
$$R_5 = \{000, 010, 020, 100, 200, 110, 120, 210, 220\}$$
 $$R_6 = \{000, 010, 020, 101, 202, 111, 121, 212, 222\} $$
 $$R_7 = \{000, 010, 020, 102, 201, 112, 122, 211, 221\}$$
 $$R_8 = \{000, 011, 022, 100, 200, 111, 122, 211, 222\}$$
 $$R_9 = \{000, 011, 022, 101, 202, 112, 120, 210, 221\}$$
 $$R_{10} = \{000, 011, 022, 102, 201, 110, 121, 212, 220\}$$
 $$R_{11} = \{000, 012, 021, 100, 200, 112, 121, 212, 221\} $$
 $$R_{12} = \{000, 012, 021, 101, 202, 110, 122, 211, 220\}$$
 $$R_{13} = \{000, 012, 021, 102, 201, 111, 120, 210, 222\}$$.

 The other blocks are additive cosets of these 13 representatives. For example, the first parallel class consist of the following blocks:
 $$B_1 = \{000, 001, 002, 010, 020, 011, 012, 021, 022\}$$
 $$B_2 = \{100, 101, 102, 110, 120, 111, 112, 121, 122\}$$
 $$B_3 = \{200, 201, 202, 210, 220, 211, 212, 221, 222\}$$
\end{example}
The overlap between blocks from different parallel classes in the case of affine resolvable designs is known from the following result.
\begin{lemma}
\label{lemma:intersect_affine}
\cite{St} Any two blocks from different parallel classes of an affine resolvable $(\theta, \rho, \alpha, \lambda)$-BIBD intersect in exactly $\alpha^2/\theta$ symbols.
\end{lemma}
\begin{corollary}
Let $P^1$ and $P^2$ be parallel classes in an affine resolvable $(\theta, \rho, \alpha, \lambda)$-BIBD. Let $\beta = \alpha^2/\theta$. Any block from $P^1$ is $\beta$-recoverable from $P^2$.
\end{corollary}
\begin{proof}
By Lemma \ref{lemma:intersect_affine} it is clear that the intersection between a block in $P^1$ and any block in $P^2$ is of size $\beta$. Next, there is no overlap between the blocks in $P^2$ and there exist $\theta/\alpha$ blocks in $P^2$; this gives us the result.
\end{proof}
Thus, for affine resolvable designs resulting from affine geometries, we have $\beta = \alpha^2/\theta = q^{m-2}$ and $d = q$.

In addition to the above examples, we emphasize that we can generate resolvable FR codes with a wide range of parameters as shown in Table \ref{table:ag_parameters}. 

\begin{table}[t]
\begin{center}
\begin{tabular}{|c| c | c | c | c |c|}
  \hline
 $m$&$\theta=q^m$&$n=q\rho$&$\alpha=q^{m-1}$&$\beta=\alpha^2/\theta$& $\rho=\frac{q^{m}-1}{q-1}$\\
  \hline
 4&16&30&8&4&15\\
 \hline
 5&32&62&16&8&31\\
 \hline
 6&64&126&32&16&63\\
 \hline
 3&27&39&9&3&13\\
 \hline
 4&81&120&27&9&40\\
 \hline
 5&243&363&81&27&121\\
 \hline
\end{tabular}
\end{center}
\caption{\label{table:ag_parameters} Parameters of REPRESSURED codes from affine geometries.}
\end{table}


For instance, let $q = 4, m = 5, d = m-1 = 4$. Then a resolvable FR code $\calC$ with $\theta = 1024, \alpha = 256, \beta = 64, \rho = 341, n = 1364$ exists. This code has $341$ parallel classes. Suppose that we wish to deploy a DSS with a repetition degree of $5$. We can simply pick $5$ parallel classes to form the node set. In the event of four node failures, we contact all the nodes in the intact parallel class and download $\beta = 64$ symbols from each of them, i.e., the code is resilient to four node failures. The code rate $R_{\calC}(k)$ is guaranteed to be at least $k \alpha - \binom{k}{2} \beta$ as any two nodes have at most $\beta$ symbols in common. This implies that these codes are universally good.


\begin{remark}
This approach provides us with a systematic way of designing codes with $\beta > 1$ that are resilient up to $\rho - 1$ failures. Furthermore, it can be seen that a system can be resilient to at most $\rho - 1$ failures, i.e., our approach is optimal from a resilience point of view.  It is known that there exist affine resolvable designs that are not Steiner systems, i.e., our class of codes is different from the Steiner system based codes considered in \cite{RR}.
\end{remark}

\subsection{Hadamard matrix based construction}
A second construction of affine resolvable designs can be obtained from Hadamard matrices or equivalently difference sets as discussed below. Consider a group $G$ of order $\theta$ and $D \subseteq G$ such that $|D| = \alpha$, with the property that every nonidentity element of $G$ can be expressed as a difference $d_1-d_2$ of elements of $D$ in exactly $\lambda$ ways. We refer to $D$ as a $(\theta,\alpha,\lambda)$-difference set.
\begin{lemma} {\it Quadratic residue difference set.} \cite{St}
Let $q=4m-1$ be an odd prime power and $G = \mathbb{F}_q$. Let $D=\{z^2: z \in \mathbb{F}_q,~z\neq0 \}=\{d_1, \cdots, d_k\}$ be the set of quadratic residues. Then $D$ is a $(4m-1, 2m-1, m-1)$-difference set in $(\mathbb{F}_q,+)$\footnote{$+$ denotes the additive operation over $\mathbb{F}_q$}.
\end{lemma}
For any $g \in G$, we define the \textit{translate} of $D$ by $g+D=\{g+d: d \in D\}$, and define the \textit{development} of $D$ by $\mbox{Dev}(D)=\{g+D: g\in G\}$. If $D$ is a $(\theta,\alpha,\lambda)$-difference set in $G$, then $(G,\mbox{Dev}(D))$ is a $(\theta, \rho,\alpha,\lambda)$-BIBD \cite{St}.

%
%

Let $(\Omega, V )$ be the $(4m-1, 2m-1, 2m-1, m-1)$-BIBD constructed by using a quadratic residue difference set. Let $\infty \notin \Omega$, and define for $V'=\{B\cup \{\infty\}: B \in V\}$. Then it can be shown \cite{St} that $(\Omega \cup \{\infty\}, V' \cup \{\Omega-B:B \in V\})$ is an affine resolvable $(4m, 4m-1, 2m, 2m-1)$-BIBD. Using the equations (\ref{bibd_eq_1}) and (\ref{bibd_eq_2}) it can be seen that this corresponds to a resolvable FR code with parameters $\theta = 4m, \alpha = 2m, \beta = m, d = 2, \rho = 4m-1$ and $n=8m-2$. 

\begin{example}
$D=\{1,2,4\}$ is a $(7, 3, 1)$-difference set in $\Omega=\bbf_7$. We can construct the Fano plane by using the difference set $D$ which is a $(7, 3, 3,1)$-BIBD. By applying the above construction we have the following parallel classes $P_i, i = 1, \dots 7$ and their corresponding storage nodes.
$$P_1 = \{ \{\infty 124\},\{0356\}\}$$
$$P_2 = \{ \{\infty235\},\{1460\}\}$$
$$P_3 = \{ \{\infty346\},\{2501\}\}$$
$$P_4 = \{ \{\infty450\},\{3612\}\}$$
$$P_5 = \{ \{\infty561\},\{4023\}\}$$
$$P_6 = \{ \{\infty602\},\{5134\}\}$$
$$P_7 = \{ \{\infty013\},\{6245\}\}$$

\end{example}
\begin{remark}
For this class of codes, $d$ is always 2. However, they offer more flexibility in the choice of $\beta$; unlike affine geometry based codes, we do not require $\beta$ to be a prime power.
\end{remark}

Table \ref{table:qrd_parameters} contains the parameters of this construction corresponding to some representative values of $m$.
\begin{table}
\begin{center}
\begin{tabular}{| c | c | c | c |c|}
  \hline
 $\theta = 4m$&$n=8m-2$&$\alpha = 2m$&$\beta = m$& $\rho = 4m-1$\\
  \hline
 8&14&4&2&7\\
 \hline
 12&22&6&3&11\\
  \hline
 20&38&10&5&19\\
  \hline
 24&46&12&6&23\\
  \hline
 28&54&14&7&27\\
  \hline
 32&62&16&8&31\\
  \hline
\end{tabular}
\end{center}
\caption{\label{table:qrd_parameters} Parameters of REPRESSURED codes from Hadamard matrices.}
\end{table}
For these codes $d=2$ which implies that $k \leq 2$. It can be seen that the code rate $R_{\calC}(k)$ is $3m$ for any Hadamard design based resolvable fractional code for $k=2$. In this case, since $3m$ is equal to $k \alpha - \binom{k}{2} \beta$, these codes are universally good.

The advantage of affine resolvable designs is that the overlap between blocks from different parallel classes is known exactly. This is not the case in general for resolvable designs that are not affine resolvable. Thus, if the design is not affine resolvable, we may not be able guarantee the $\beta$-recoverable property of the fractional codes. However, it is conceivable that general resolvable designs can be used for the design of FR codes. In the next section we present constructions of resolvable FR codes where $\beta = 1$, but the corresponding designs are not affine resolvable. We emphasize that these designs result in FR codes whose parameters are not achieved by \cite{RR}.

\section{REPRESSURED codes from Latin Squares}
\label{sec:general_resolv}

In this section, we consider resolvable FR codes with $\beta = 1$. We use mutually orthogonal Latin squares \cite{St} for the construction.
\begin{definition}A Latin square of order $s$ with entries from a $s$-set $\Omega$ is an $s \times s$ array $L$ in which every cell contains an element of $\Omega$ such that every row of $L$ is a permutation of $\Omega$ and every column of $L$ is a permutation of $\Omega$.
\end{definition}
\begin{definition} Suppose that $L_1$ and $L_2$ are Latin squares of order $s$ with entries from $\Omega_1$ and $\Omega_2$ respectively (where $|\Omega_1| = |\Omega_2|$). We say that $L_1$ and $L_2$ are orthogonal Latin squares if for every $x \in \Omega_1$ and for every $y \in \Omega_2$ there is a unique cell $(i,j)$ such that $L_1(i,j)=x$ and $L_2(i,j)=y.$
\end{definition}
Equivalently, one can consider the superposition of $L_1$ and $L_2$ in which each cell $(i,j)$ is occupied by the pair $(L_1(i,j), L_2(i,j))$. Then, $L_1$ and $L_2$ are orthogonal if and only if the resultant array has every value in $\Omega_1 \times \Omega_2$. A set of $r$ Latin squares $L_1, \dots, L_r$ of order $s$ are said to be mutually orthogonal if $L_i$ and $L_j$ are orthogonal for all $1 \leq i < j \leq r$.

We now demonstrate a procedure of constructing FR codes from mutually orthogonal Latin squares \cite{Yates}.
Let $\Omega=\{1, 2, \cdots, s^2\}$, and let $L_1, L_2, \cdots L_{r-2}$ be a set of $r-2$ mutually orthogonal Latin squares of order $s$ ($r-2 \leq s-1$).
\begin{itemize}
\item Arrange the elements of $\Omega$ in a $s \times s$ array $A$. Each row and each column of $A$ corresponds to a storage node (this gives us $2s$ nodes).
\item Note that  $L_i$ takes values in $\{1, \dots, s\}$. Within $L_i$ identify the set of $(i,j)$ pairs where a given value $z \in \{1, \dots, s\}$ appears. Create a storage node by including the entries of $A$ corresponding to the identified $(i,j)$ pairs.
\item Repeat this for each $L_i$ and all $z \in \{1, \dots, s\}$. This creates another $(r-2)s$ storage nodes.
\end{itemize}
Thus, a total of $rs$ storage nodes of size $s$ can be obtained. Of course one can choose fewer storage nodes if so desired.
\begin{lemma}
The construction procedure described above produces a resolvable fractional repetition code with $\theta = s^2, n=rs, d=\alpha = s, \rho = r$ and  $\beta = 1$.
\end{lemma}
\begin{proof}
It is clear from the construction that $\theta = s^2$ and $n = rs$. Each storage node has $s$ symbols so that $\alpha = s$. We need to show that the code is resolvable. Towards this end note that it is evident that we obtain a parallel class by considering the nodes corresponding to the rows of $A$ (similar argument holds for the columns of $A$). Next the nodes obtained by considering Latin square $L_i$ also form a parallel class, since the set of elements obtained by considering the $(i,j)$ pairs corresponding to $z_1 \in \{1, \dots, s\}$ are distinct from those corresponding to  $z_2 \in \{1, \dots, s\}$, if $z_1 \neq z_2$. As we have $r$ parallel classes, we obtain $\rho = r$. Next, consider the overlap between any two storage nodes belonging to different parallel classes. As $L_i$ and $L_j$ are orthogonal, any entry $(k,l) \in [s]\times[s]$ appears exactly once in the superposition of $L_i$ and $L_j$, which implies that the overlap between storage nodes from different parallel classes corresponding to the $L_i$'s is exactly one element. Similarly, a block from a parallel class corresponding to $L_i$ has exactly one overlap with the blocks corresponding to the rows and columns of $A$.
\end{proof}
\begin{example} \label{latin}
Let $s=4$, and $r = 2$. Then, we have the following construction
$$A=\begin{array}{cccc}
1&2&3&4\\
5&6&7&8\\
9&10&11&12\\
13&14&15&16
\end{array}$$ ,
$$
L_1=\begin{array}{cccc}
1&2&3&4\\
2&1&4&3\\
3&4&1&2\\
4&3&2&1
\end{array} \mbox{~and}~
 L_2=\begin{array}{cccc}
1&2&3&4\\
3&4&1&2\\
4&3&2&1\\
2&1&4&3
\end{array} $$
where it can be verified that $L_1$ and $L_2$ are orthogonal.
Then we have the following parallel classes and corresponding storage nodes.
\small
$$P^{\text{rows}} = \{\{1,  2,  3,  4\},  \{5,  6,  7,  8\}, \{9,  10,  11,  12\},  \{13, 14,  15,  16\}\} $$
$$P^{\text{cols}} = \{\{1,  5,  9,  13\},  \{2,  6,  10,  14\},   \\ \{3,  7,  11,  15\},  \{4,  8,  12,  16\}\}$$
$$P^{L_1} = \{\{1,  6,  11,  16\},  \{2,  5,  12,  15\},  \\  \{3,  8,  9,  14\},  \{4,  7,  10,  13\}\}$$
$$P^{L_2} = \{\{1,  7,  12,  14\},  \{2,  8,  11,  13\},   \\ \{3,  5,  10,  16\},  \{4,  6,  9,  15\}\}$$

\end{example}

Note that in describing the above construction we assumed the existence of $r-2$ mutually orthogonal Latin squares. We now discuss the issue of the existence of such structures.

If $p$ is a prime number, $s$ is a positive integer, and $N = p^s$ then we can construct $N-1$ mutually orthogonal Latin squares as described below.
\begin{itemize}
\item[(i)] Define $L_a : \bbf_N \times \bbf_N \rightarrow \bbf_N$, by $(r, c) \mapsto ar + c$ (where the addition is over $\bbf_N$)
for all $a\in \bbf_N \setminus \{0\}$. Then $L_a$ is a Latin square since for a given row $r$ (or column $c$) the column (or row) location of an element $s$ is uniquely specified. 
\item [(ii)] For any $a, b \in \bbf_N \setminus \{0\}$, $L_a$ and $L_b$ are orthogonal since for given ordered pair $(s, t)$ the system $ar+c = s$, $br+ c = t$, determine $r = (a- b)^{-1}(s-t)$ and $c = s-ar$ uniquely.
\end{itemize}
\begin{example}
Let N=3. Then  $\mathbb{F}_3=\{0,1,2\}$, $L_1: x+y$ and $L_2: 2x+y$. The two orthogonal Latin squares of order 3 constructed by the above method are
$$L_1=\begin{array}{ccc}
0&1&2\\
1&2&0\\
2&0&1\\
\end{array},
~L_2=\begin{array}{ccc}
0&1&2\\
2&0&1\\
1&2&0\\
\end{array}$$
\end{example}
In general, the construction of orthogonal Latin squares is somewhat involved. However, the celebrated results of \cite{BSP}, demonstrate the construction of two orthogonal Latin squares for all orders $N \neq 2, 6$.
This immediately allows us to construct resolvable fractional repetition codes with the following parameters $n=4s, \theta=s^2,$ $d=\alpha=s,$ $\beta=1,$ and $\rho=4$ for any $s\neq 2, 6$.

It can be observed that this construction allows us to design FR codes that are not covered by those arising from Steiner systems. For instance, Let $\alpha=10$ and $\theta=100$. Then to construct a FR code we need use the Steiner system $S(2,10,100)$ which does not exist \cite{Lam}. However the above construction with two orthogonal Latin squares of order 10 provides us a resolvable fractional code with $\alpha=10$ and $\theta=100$.

\section{Concluding Remarks}
\label{sec:conc_remarks}
In this work we introduced REPRESSURED codes that are fractional repetition codes constructed from resolvable designs. 
Our work offers the following advantages with respect to the existing work.
\begin{itemize}
\item In \cite{RR}, they only considered FR codes with $\beta = 1$, i.e., codes where the new node downloads exactly one packet from the survivor nodes. In contrast, our constructions based on affine geometries and Hadamard designs (cf. Section \ref{sec:affine_resolv}), allow for a large family of codes where $\beta > 1$.
\item The resolvable nature of our codes allows for a natural tradeoff between the repetition degree $\rho$ and the number of parallel classes, i.e., we can obtain FR codes with higher or lower $\rho$ simply by including or removing parallel classes. This flexibility is lacking in the approach of \cite{RR}, where the entire Steiner system needs to be used. In particular, it is known that there exist Steiner systems that are not resolvable. For instance, Ray-Chaudhuri and Wilson showed that a resolvable $S(2,\alpha,\theta)$ exists if and only if  $\theta \equiv 3\mod{6}$ (\cite{CW}).
\item One of the open questions of \cite{RR} was the existence of FR codes with $\rho > 2$ that were not covered by Steiner systems. Our constructions in Section \ref{sec:general_resolv} provide such examples.
\item The work of \cite{kooG11} mentioned the usage of affine planes for designing distributed storage systems. Our work has considered a larger class of designs based on affine geometries, Hadamard designs and Latin squares.
\end{itemize}
It is to be noted that all our proposed constructions are universally good, i.e. $R_\calC (k) \geq k \alpha - \beta \binom{k}{2}$. Moreover, the repair process is particularly simple. With $\rho - 1$ failures, at least one parallel class is guaranteed to be intact. The new node can simply contact the nodes in the intact parallel class for regeneration. This property is likely to simplify the implementation of the proposed systems.

Future work would include the investigation of other classes of combinatorial designs and a more careful analysis of the maximum filesize of the proposed codes.

\bibliographystyle{IEEETran}

\end{document}